\documentclass[copyright,creativecommons]{eptcs}
\providecommand{\event}{GaSCOM 2024} 

\usepackage{iftex}

\ifpdf
  \usepackage{underscore}         
  \usepackage[T1]{fontenc}        
\else
  \usepackage{breakurl}           
\fi

\usepackage{enumitem,stix}

\usepackage{hyperref}
\usepackage[utf8]{inputenc}
\usepackage{amsmath,amsthm,stix}
\usepackage{stmaryrd}
\usepackage{comment}
\usepackage{calligra}
\allowdisplaybreaks
\usepackage{xcolor}
\usepackage{changepage}
\usepackage{etoolbox}

\makeatletter
\patchcmd{\upbracefill}{\m@th}{\scriptstyle\m@th}{}{}
\patchcmd{\upbracefill}{$\braceld$}{$\scriptstyle\braceld$}{}{}
\patchcmd{\upbracefill}{\bracelu}{\bracelu\mkern-1mu}{}{}
\patchcmd{\upbracefill}{\hfill\braceru}{\hfill\mkern-1mu\braceru}{}{}
\makeatother

\def\sttf2#1#2{\left[\!\!\left[#1\atop#2\right]\!\!\right]}
\def\stss2#1#2{\left\{\!\!\left\{#1\atop#2\right\}\!\!\right\}}

\DeclareMathAlphabet{\mathcalligra}{T1}{calligra}{m}{n}
\DeclareMathAlphabet{\mathgoth}{OT1}{goth}{m}{n}

\newtheorem{theorem}{Theorem}[section]
\newtheorem{definition}[theorem]{Definition}
\newtheorem{thm}[theorem]{Theorem}
\newtheorem{prop}[theorem]{Proposition}

\newtheorem{lemma}[theorem]{Lemma}
\newtheorem{remark}[theorem]{Remark}

\title{Type-$B$ analogue of Bell numbers\\ using Rota's Umbral calculus approach}

\author{
Eli Bagno\\
\institute{Jerusalem College of Technology,\\ 21 HaVaad HaLeumi St.,\\ Jerusalem, Israel,\\ and  Michlalah College Jerusalem,\\ 36 Barukh Duvdevani\\ St., Jerusalem, Israel}
\email{bagnoe@g.jct.ac.il}
\and
\medskip
David Garber\\
\institute{Holon Institute of Technology,\\ 52 Golomb St., P.O.Box 305,\\ 5810201 Holon, Israel}
\email{garber@hit.ac.il}
}
\def\titlerunning{Type-$B$ analogue of Bell numbers using Rota's Umbral calculus approach}
\def\authorrunning{Eli Bagno and David Garber}

\begin{document}

\maketitle

\begin{abstract}
Rota used the functional $L$ to recover old properties and obtain some new formulas for the Bell numbers.
Tanny used Rota's functional $L$ and the celebrated Worpitzky identity to obtain some expression for the ordered Bell numbers, which can be seen as an evident to the fact that the ordered Bell numbers are gamma-positive. In this paper, we extend some of Rota's and Tanny's results to the framework of the set partitions of Coxeter type $B$.
\end{abstract}

\section{Introduction}
Rota \cite{Rota} declares:

\begin{adjustwidth}{1cm}{}
``It is the author's conviction that formula
(4), which we derive below, is the natural description of the exponential numbers. The basic idea is a general one, and can be applied to a variety of other
combinatorial investigations. We shall see that it easily leads to quick derivations of the properties of the $B_n$.''
\end{adjustwidth}

The `exponential numbers' mentioned in the cited paragraph are the obsolete name for what we call today {\it Bell numbers}, which count set partitions of the set $[n]=\{1,\dots,n\}$. Rota's
Formula (4) reads: $B_n=L(u^n)$, where $L$ is a linear functional defined on the vector space of polynomials in the indeterminate $u$ and $B_n$ is the $n$-th Bell number.

Rota \cite{Rota} used the functional $L$ to recover old properties and obtain some new formulas for the Bell numbers.
Explicitly, let $V = \mathbb{R}[u]$ be the vector space of all real polynomials
in the single variable $u$. Then any sequence of polynomials of degree $0, 1, 2, \dots $ is a basis for this vector space, in particular the sequence of the {\it falling factorials} $(u)_0 = 1, (u)_1 = u, (u)_2 = u(u- 1),\cdots$ is a basis as well.  Let $L : \mathbb{R}[u] \to \mathbb{R}$ be the linear functional that is uniquely defined by
$L((u)_j) = 1$,
for all $j \geq 0$. Rota \cite{Rota} states the following theorem:
\begin{thm}\label{theorem of Rota}
Let $n \geq 0$. Then:\\
(1) $B_n = L(u^n)$,\\
(2) $B_{n+1} = \sum\limits_{j=0}^n {n \choose j} B_j$, \\
(3) $B_n = \frac{1}{e} \sum\limits_{j \geq 0} \frac{j^n}{j!}$.
\end{thm}

Tanny \cite{Tanny} used Rota's functional $L$ and the celebrated Worpitzky identity $m^n=\sum\limits_{k=0}^n{{n+m-k \choose {n}}}a_{n,k},$ to obtain the following expression for the ordered Bell numbers :
$$F_n(x)=\sum\limits_{k=1}^n a_{n,k}x^{n-k+1}(1+x)^{k-1},$$ where $F_n(x)$ is the generating function of the number of ordered set partitions of the set $[n]$ and $a_{n,k}$ are Eulerian numbers, a.k.a. the number of permutations in the symmetric group $S_n$ having $k$ descents, see \cite{Pet-Book}.
In modern terms, this expression can be seen as an evident to the fact that the ordered Bell numbers are gamma-positive.

\medskip

In this paper, we extend some of Rota's and Tanny's results to the framework of the set partitions of Coxeter type $B$.

In Section \ref{Background}, we recall the definition of set partitions (and ordered set partitions) of type $B$ and introduce the definition of Bell numbers and Bell polynomials of type $B$. In Section \ref{sec Rota type B}, we extend the first two parts of Rota's Theorem \ref{theorem of Rota} to Bell numbers of type $B$. Actually part $(2)$ of that result will be proven both by Rota's method and by a counting argument.
In Section \ref{sec Tanny gamma positivity type B}, we use Brenti's generalization of Worpitzky's identity to obtain a gamma-positivity result for the ordered Bell polynomials of type $B$.

\section{Set partitions and Bell numbers of type
{\it B}} \label{Background}

\subsection{Set partitions of type {\it B}}
We now recall the definition of set partitions of type $B$ (see Dolgachev-Lunts \cite[p.~755]{DoLu} and Reiner \cite[Section 2]{R}; mentioned implicitly in Dowling \cite{Dow} and Zaslavsky \cite{Za} in the form of signed graphs):

\begin{definition}\label{def of type B par}
Denote: $[\pm n]:=\{\pm1,\dots,\pm n\}$. A {\em set partition of $[n]$ of type $B$} or a {\em signed set partition} is a set partition of  the set $[\pm n]$  such that the following conditions are satisfied:
\begin{itemize}
\item If $B$ appears as a block in the set partition, then $-B$ (which is obtained from $B$ by negating all its elements) also appears in that partition.
\item There exists at most one block satisfying $-B=B$. This block is called the {\em zero block} (if it  exists, it is a subset of $[\pm n]$ of the form $\{\pm i \mid i \in C\}$ for some $C \subseteq [n]$).
\end{itemize}
\end{definition}

For example, the following is a set partition of $[6]$ of type $B$:
\begin{equation*}\label{example of standard presentation}
\{\{1,-1,4,-4\},\{2,3,-5\},\{-2,-3,5\},\{6\},\{-6\}\}.
\end{equation*}

Note that every non-zero block $B$ has a corresponding block $-B$ attached to it. For the sake of convenience, we write for the pair of blocks $B,-B$, only the representative block containing the minimal {\it positive} number appearing in $B \cup -B$.
For example, the pair of blocks
$B=\{-2,-3,5\},-B=\{2,3,-5\}$ will be represented by the single block $\{\mathbf{2},3,-5\}$.

Our convention will be to write first the zero block and denote it by $B_0$ if exists and then the non-zero blocks of a set partition of type $B$ in such a way that the sequence of absolute values of the minimal elements of the blocks is increasing. We call this the {\it standard presentation}.

For example, the following is a set partition of $[6]$ of type $B$ in its standard presentation:
$$\left\{B_0=\{1,-1,4,-4\},B_1=\{\mathbf{2},3,-5\},B_2=\{\mathbf{6}\}\right\}.$$

\subsection{Stirling numbers and Bell numbers of type {\it B}}

\begin{definition} \label{definition of Stirling number of type B second kind}
Let $S^B(n,k)$ be the number of set partitions of type $B$ having $k$ representative non-zero blocks. This is known as {\em the Stirling number of type $B$ of the second kind} (see sequence A085483 in OEIS \cite{OEIS}).
\end{definition}

It is easy to see that $S^B(n,n)=S^B(n,0)=1$ for each $n\geq 0$.  The following recursion for $S^B(n,k)$ is well-known (\cite[Theorem 7; see the Erratum]{Dow}, \cite[Corollary 3]{Beno}, for $m=2$, and \cite[Equation (1)]{Wa}, for $m=2,c=1$):

\begin{prop}\label{recursion theorem}
For each $1 \leq k < n$,
\begin{equation}
S^B(n,k)=S^B(n-1,k-1)+(2k+1)S^B(n-1,k).
\label{recursion_second_kind_type_b}
\end{equation}
\end{prop}

The following result was proved combinatorially in \cite{monthly} using a `balls into urns' approach:

\begin{thm}\label{balls to urns}
Let $x \in \mathbb{R}$ and let $n \in \mathbb{N}$. Then we have:
\begin{equation}
x^n=\sum\limits_{k=0}^n S^B (n,k)(x)^B_k,\label{B}
\end{equation}
where $(x)^B_k:=(x-1)(x-3)\cdots (x-2k+1)$ and $(x)^B_0:=1$, called {\em the falling factorial of type $B$}.
\end{thm}

\begin{remark}\label{connection}
There is a simple connection between the falling factorials of types $A$ and $B$:
$$\left(\frac{x-1}{2}\right)_n= \left(\frac{x-1}{2}\right)\left(\frac{x-1}{2}-1\right)\cdots \left(\frac{x-1}{2}-n+1\right)=\frac{1}{2^n}(x-1)(x-3)\cdots (x-2n+1)=\frac{1}{2^n}(x)_n^B,$$
where {\em the falling factorial of type $A$} is defined as follows:
$(x)_k:=x(x-1)(x-2)\cdots (x-k+1)$.
\end{remark}

\medskip

We define the {\it Bell number of type $B$} as follows:
$B_n^B=\sum\limits_{k=0}^nS^B(n,k).$
Obviously, the Bell number counts all the set partitions of the set $[ n]$ of type $B$. A similar definition appears in Mez\H{o} and Ram\'irez \cite{MeRa}.

\medskip

We define also the {\it Bell polynomial of type $B$} as follows:
$B_n^B(u)=\sum\limits_{k=0}^nS^B(n,k)u^k.$

\medskip

Sagan and Swanson \cite{SaSw} discussed {\it  ordered set partitions of type $B$}, which are defined as follows:

\begin{definition}
An {\em ordered set partition of $[n]$ of type $B$ having $k$ non-zero blocks} is a sequence of sets $(S_0,S_1,\dots, S_{2k+1})$ which form a set partition of $[n]$ of type $B$, such that the following two order conditions are satisfied:

\begin{enumerate}
\item For each $1\leq i \leq n$, $i\in S_0$ if and only if $-i\in S_0$  (i.e. the zero block $S_0$ is always the first block).
\item For each $1 \leq j \leq k$, we have $S_{2j}=-S_{2j-1}$.
\end{enumerate}
\end{definition}

Similar to the ordinary Bell number of type $B$, we define the {\it ordered Bell number of type $B$} and its associated {\it ordered Bell polynomial} as follows:

\begin{definition}
 $B_n^{B,o}=\sum\limits_{k=0}^n { 2^k k! S^B(n,k)}$ and
 $B_n^{B,o}(x)=\sum\limits_{k=0}^n {2^k k! S^B(n,k)x^k}$
\end{definition}

\section{Type-{\it B} analogue of Rota's result}\label{sec Rota type B}

Let $V$ be the vector space of all polynomials in the
variable $x$. It is easy to see that the set
$\left\{(x)^B_k\right\}_{k\in \mathbb{Z}}$ is a basis of $V$. We use this basis to define a functional $L:V \rightarrow \mathbb{R}$ by
$L((x)_k^B)=u^k$ for each $k\in \mathbb{Z}$, where $u$ is a fixed real number.

\medskip

We will use the following lemma in the proof of Theorem \ref{L(x^n)}(2):

\begin{lemma}\label{property of L}
Let $p(x)$ be any polynomial. Then:
$L((x-1)p(x-2))=u\cdot L(p(x)).$
\end{lemma}

\begin{proof}
Since the set of falling factorials of type $B$\\
\centerline{$\{1,(x-1),(x-1)(x-3),\dots,(x-1)(x-3)\cdots (x-2n+1)\}$}
is a basis for $\mathbb{R}_n[x]$, we can write any polynomial $p(x)$ as a linear combination of this set as follows:
\begin{equation}\label{RHS}
p(x)=a_0+a_1(x-1)+a_2(x-1)(x-3)+a_n(x-1)(x-3)\cdots (x-2n+1).
\end{equation}
Substituting $x-2$ for $x$, we have now:\\
\centerline{$p(x-2)=a_0+a_1(x-3)+a_2(x-3)(x-5)+\cdots+a_n(x-3)(x-5)\cdots (x-2n+3).$}
Multiplying the last equality by $(x-1)$ yields:
\begin{equation}\label{LHS}
(x-1)p(x-2)=a_0(x-1)+a_1(x-1)(x-3)+\cdots+ a_n(x-1)\cdots (x-2n+3).
\end{equation}
Applying the operator $L$ on Equations (\ref{RHS}) and (\ref{LHS}), we get the desired result.
\end{proof}

We present here the analogue for type $B$ of the celebrated results by Rota \cite{Rota} (for part (2), see also  Mez\H{o} and Ram\'irez \cite[p. 258]{MeRa}):

\begin{thm}\label{L(x^n)}
For each $n\in \mathbb{N}$, we have:\\
(1) $B_n^B(u)=L(x^n)$.\\
(2) $B_{n+1}^B(u)=B_n^B(u)+ u \cdot \sum\limits_{j=0}^n{2^{n-j}{n \choose j} B_j^B(u)}.$
\end{thm}

\begin{proof}
(1) By Theorem \ref{balls to urns}, we have $x^n=\sum\limits_{P \in \Pi_n^B}(x)_{N(P)}^B$, where $\Pi_n^B$ is the set of set partitions of the set $[n]$ of type $B$ and $N(P)$ is the number of non-zero representative blocks in the set partition $P$. Now apply the functional $L$ on both sides of this equation and use the linearity of $L$ to get $$L(x^n)=\sum\limits_{P \in \Pi_n^B}L\left((x)_{N(P)}^B\right)=\sum\limits_{P \in \Pi_n^B}u^{N(P)}=\sum\limits_{k=0}^n{S^B(n,k)u^k}=B_n^B(u).$$

\noindent
(2) For this result, we supply two different proofs: an algebraic one and a combinatorial one for $u=1$.
We start  with the algebraic proof:
Applying Lemma \ref{property of L} to $p(x)=(x+2)^n$, we get:
\begin{eqnarray*}
B_{n+1}^B(u)-B_n^B(u)& = &L(x^{n+1})-L(x^n)=L((x-1)x^n) \stackrel{{\rm Lemma} \ \ref{property of L}}{=} u \cdot L((x+2)^n)=\\
&=&u \cdot L\left(\sum\limits_{j=0}^n {n \choose j} x^j 2^{n-j}\right)=u \cdot \sum\limits_{j=0}^n {n \choose j} 2^{n-j}L(x^j)=u \cdot \sum\limits_{j=0}^n {n \choose j} 2^{n-j}B_j^B(u).
\end{eqnarray*}
The combinatorial proof for $u=1$ is as follows:
The left hand side counts the total number of set partitions of the set $[n]$ of type $B$ for any number of non-zero representative blocks. we show that the right hand side counts the same thing in a different way.
We divide in two cases according to the location of $n+1$:
\begin{enumerate}
    \item If $n+1$ is located in the zero-block, then we have $B_n^B$ possibilities to locate the other elements.
    \item Otherwise, $n+1$ is located in a non-zero block. Then, for each $0\leq k\leq n$, assume that $n-k$ is the number of elements that share a block with $n+1$. The rest $k$ elements can be located in $B_k^B$ ways. Finally, we have $2^{n-k}$ possibilities to sign the elements in the block containing $n+1$.
\end{enumerate}

\end{proof}

\section{Gamma-positivity of ordered Bell polynomial of type {\it B}}\label{sec Tanny gamma positivity type B}

Let $V$ be the vector space of all polynomials in the
variable $x$. It is easy to see that the set
$\left\{(x)^B_k\right\}_{k\in \mathbb{Z}}$ is a basis of $V$. We use this basis to define a new functional $L^o:V \rightarrow \mathbb{R}$ by
\begin{equation}
L^o((x)_k^B)=2^kk!u^k
\label{def ordered L}\end{equation}
for each $k\in \mathbb{Z}$, where $u$ is a fixed real number.
By slightly modifying the proof of Theorem \ref{L(x^n)} we get that $L^o(x^n)=B_n^{B,o}(u)$.

Brenti \cite{Brenti} obtained a Worpitzky-like identity for type $B$ as follows:
For each $m,n\in \mathbb{N}$, one has\\
\centerline{$(1+2m)^n=\sum\limits_{k=0}^n{{n+m-k \choose {n}}}E^B_{n,k},$}
where $E^B_{n,k}$ are the Eulerian numbers of type $B$, which counts the number of signed permutations in the Coxeter group of type $B$ having $k$ descents;
this set of numbers constitutes the sequence A060187 in OEIS \cite{OEIS}. For a combinatorial proof of this identity, see \cite{BGN}.

\medskip

The following result shows the gamma-positivity of the ordered Bell polynomials of type $B$:
\begin{prop}
$B_n^{B,o}(u)=\sum\limits_{k=0}^n E^B_{n,k} u^k (1+u)^{n-k}$.
\end{prop}

\begin{proof}
If we write $x=1+2m$, then we have:
\begin{eqnarray*}
B_n^{B,o}(u)&=&L^o(x^n) \quad = \quad \sum\limits_{k=0}^n L^o\left[{n+m-k \choose {n}}\right]E^B_{n,k}=\\
& = & \sum\limits_{k=0}^n L^o\left[{n+\frac{x-1}{2}-k \choose {n}}\right]E^B_{n,k}\quad = \quad \sum\limits_{k=0}^n \frac{E^B_{n,k}}{n!} L^o\left[\left(n+\frac{x-1}{2}-k\right)_n\right]=\\
& \stackrel{(*)}{=} & \sum\limits_{k=0}^n \frac{E^B_{n,k}}{n!} L^o\left[\sum\limits_{\ell=0}^n {n \choose \ell}(n-k)_{n-\ell}\left(\frac{x-1}{2}\right)_\ell \right]=\\
& \stackrel{{\rm Rem.} \ \ref{connection}}{=} & \sum\limits_{k=0}^n \frac{E^B_{n,k}}{n!} \sum\limits_{\ell=0}^n {n \choose \ell}(n-k)_{n-\ell}L^o\left[\frac{(x)^B_{\ell}}{2^\ell}\right] =\\
& \stackrel{{\rm Eqn.} \ (\ref{def ordered L})}{=} & \sum\limits_{k=0}^n E^B_{n,k} \sum\limits_{\ell=0}^n {n \choose \ell}\frac{(n-k)_{n-\ell}}{n!}\ell! u^{\ell}
\quad = \quad \sum\limits_{k=0}^n E^B_{n,k} \sum\limits_{\ell=0}^n {n-k \choose n-\ell}u^{\ell} = \\
& = &  \sum\limits_{k=0}^n E^B_{n,k}u^k \sum\limits_{\ell=0}^n {n-k \choose \ell-k}u^{\ell-k}  \ \ =  \ \  \sum\limits_{k=0}^n E^B_{n,k} u^k (1+u)^{n-k},
\end{eqnarray*}

\noindent
where Equality $(*)$ is based on the binomial Umbral identity of the falling factorials of type $A$:\break $(a+b)_n=\sum\limits_{k=0}^n {n \choose k} (a)_k (b)_{n-k}$, see e.g. Roman \cite[p. 29]{Roman}.
\end{proof}

\bibliographystyle{eptcs}
\bibliography{rota}

\end{document}